\documentclass[a4paper,10pt]{article}

\title{Coding theory for noiseless channels realized by anonymous oblivious mobile robots} 

\author{Yukiko Yamauchi\thanks{Corresponding author. Faculty of Information Science and Electrical Engineering, Kyushu University, Japan. E-mail: \texttt{yamauchi@inf.kyushu-u.ac.jp}} \and 
Masafumi Yamashita\thanks{Faculty of Information Science and Electrical Engineering, Kyushu University, Japan. E-mail: \texttt{mak@inf.kyushu-u.ac.jp}}}

\usepackage{amsmath,amsfonts} 
\usepackage{version}
\usepackage{algorithmic, algorithm}
\usepackage{theorem}
\usepackage[pdftex]{graphicx}
\usepackage{latexsym}
\usepackage{subfig}

\theoremstyle{plain}
\newtheorem{theorem}{Theorem}
\newtheorem{lemma}[theorem]{Lemma}
\newtheorem{corollary}[theorem]{Corollary}
\newtheorem{observation}[theorem]{Observation}

\newenvironment{proof}{{\bf Proof. } }

\newcommand{\Prob}{\mathrm{P}}

\begin{document}
\date{}
\maketitle

\begin{abstract}
We propose an information transmission scheme by a swarm 
of anonymous oblivious mobile robots on a graph. 
The swarm of robots travel from a sender vertex 
to a receiver vertex to transmit a symbol generated at the sender. 
The codeword for a symbol is a pair of an initial configuration 
at the sender and a set of terminal configurations at the receiver. 
The set of such codewords forms a code. 
We analyze the performance of the proposed scheme 
in terms of its code size and transmission delay. 
We first demonstrate that a lower bound of the transmission delay 
depends on the size of the swarm, 
and the code size is upper bounded by an exponent of the size of the swarm. 
We then give two algorithms for a swarm of a fixed size.
The first algorithm realizes a near optimal code size with a large transmission delay.
The second algorithm realizes an optimal transmission delay with a smaller code size. 
We then consider information transmission by swarms of different sizes 
and present upper bounds of the expected swarm size by the two algorithms. 
We also present lower bounds by Shannon's lemma and noiseless coding theorem. 


\end{abstract}

\section{Introduction}

Memory is indispensable to a computer system to demonstrate 
its computation ability.
Its importance does not diminish in a distributed system.
We thus tend to guess that a distributed system lacking memory 
can solve no problems but trivial ones.
A \emph{swarm} of anonymous oblivious mobile robots is a distributed 
system whose components
called robots move on a continuous space or a graph.
It is typically characterized by the lack of identifiers 
and common and local memories.
Contrary to the quess, in spite of the absence of memory,
it has been shown to have rich ability to solve a variety of problems 
(e.g., \cite{BMPT11,CFPS12,CDPIM10,DFSY15,DLPRT12,DFSV18,DP07,DPV10,FIPS13,FPSW08,FYOKY15,LYKY18,PPV15,SY99,YS10,YUKY17,YUKY16,YY14}).
In order to solve those problems, 
the robots indeed need to ``remember'' key information 
to solve a problem 
such as the current search direction in exploration \cite{BMPT11,DLPRT12,DYKY18,FIPS13}, 
the currently agreed common coordinate system in pattern 
formation \cite{FYOKY15,SY99,YS10,YY14,YUKY16,YUKY17}, 
and the current phase in forming a sequence of patterns \cite{DFSY15}.
A main idea to make up the lack of memory is to use ``external memory''
composed of the locations of all robots, 
i.e., its (global) configuration,
and the information that the robots need to remember is 
embedded in the current configuration.
However for the robots to keep external memory stable 
is by no means easy,
and interesting tricks have been developed to realize the idea
based on concepts such as 
the smallest enclosing circle \cite{DPV10,FPSW08,FYOKY15,LYKY18,SY99,YS10,YY14}, 
the rotation group \cite{YUKY17,YUKY16} 
and the Fermat point \cite{CFPS12}.
Proposing a suitable external memory is thus considered 
to be a key in the design of algorithm 
for a swarm of anonymous oblivious mobile robots. 

This paper investigates an information transmission problem on a graph, 
which asks an algorithm to transmit information from a sender 
to a receiver 
by using a swarm of anonymous oblivious mobile robots 
moving on the graph, 
and analyzes the bounds on the amount of information 
that the swarm of robots can carry in its external memory 
and the transmission delay.

Let $G = (V,E)$ be a connected undirected (possibly infinite) graph 
with two distinguished vertices $u_S$ and $u_R$ called the 
\emph{sender} and the \emph{receiver}, respectively.
We want to transmit a symbol $s$ in 
$S = \{ s_1, s_2, \ldots , s_{\alpha} \}$ 
from $u_S$ to $u_R$ by using a swarm of anonymous oblivious 
mobile robots on $G$, 
where the probability that $s = s_i$ is $p_i$.
Sometimes we assume that the distance between $u_S$ and $u_R$ 
is sufficiently larger than $\alpha$ and/or the number $k$ of the robots
to reduce possible disturbance caused by boundary conditions in analyses.

Each vertex of $G$ can accommodate at most one robot of the swarm. 
Thus the \emph{configuration} of the swarm is represented by a subset 
$U$ of $V$.
Any configuration of the swarm must be connected; 
${\cal C} = \{ U \mid G[U] \text{ is connected} \}$ 
denotes the set of all connected configurations,
where $G[U]$ is the subgraph of $G$ induced by $U$.
When the current configuration is $C$,
a robot $r$ at $u \in C$ can move to one of its neighbors $v \not\in C$ 
if the next configuration $(C \setminus \{u\}) \cup \{v\}$ 
is also connected.
At most one robot can move at one time step.
The \emph{behavior} of the swarm of robots is determined by a 
deterministic algorithm,
which specifies the next configuration 
by choosing a robot to move and its destination. 
The input of the algorithm is the current configuration. 
Thus the behavior of the swarm is always deterministic.
A configuration $C$ is said to be \emph{initial} (resp. \emph{terminal}) 
if $u_S \in C$ (resp. $u_R \in C$).
The sets ${\cal C}_I = \{ C \in {\cal C} \mid u_S \in C \}$ and 
${\cal C}_T = \{ C \in {\cal C} \mid u_R \in C \}$ 
denote the sets of initial and terminal configurations, respectively.

Let ${\cal A} = \{ A_k \mid k = 1, 2, \ldots \}$ be a set of algorithms,
where $A_k$ is an algorithm for the swarm of $k$ robots.
For an initial configuration $C_0 \in {\cal C}_I$, where $|C| = k$,
the behavior $C_0, C_1, \ldots$ of the swarm of $k$ robots 
under $A_k$ may eventually 
reach $C_d \in {\cal C}_T$ for the first time.
In this case we define that $\tau(C_0) = C_d$ and the transmission delay is $d$.
Otherwise, $\tau(C_0) = \bot$.
Suppose that there are $\alpha$ initial configurations 
$C_0^1, C_0^2, \ldots , C_0^{\alpha}$ 
such that $\tau(C_0^i) \not= \bot$ for $i = 1, 2, \ldots , \alpha$
and $\tau(C_0^i) \not= \tau(C_0^j)$ for all $1 \leq i < j \leq \alpha$.
Assuming that both $u_S$ and $u_R$ are accessible to the list of pairs $(C_0^i, \tau(C_i^0))$ 
for $i = 1, 2, \ldots , \alpha$,
we can transmit a symbol $s_i \in S$ as follows:
\begin{enumerate}
 \item 
The sender initializes the swarm of $k_i$ robots by moving the robots to
the vertices in $C_0^i$ and starts algorithm $A_{k_i}$, 
where $k_i = |C_0^i|$.
\item
The receiver recognizes $\tau(C_0^i)$ when a robot reaches $u_R$,
and knows that the sender transmitted $s_i$.
\end{enumerate}
Here we assume that the sender has a bag of infinitely many robots 
and can initialize the swarm of any number of robots, 
and the receiver recognizes $\tau(C_0^i)$ 
as soon as a robot reaches $u_R$.
In this paper, we fix 
this information transmission scheme and
analyze its performance.

Given a set of algorithms $\cal A$,
the performance of the scheme mainly depends on the selection of 
an initial configuration $C_0^i$ for each $s_i$.
A simple approach is to choose configurations with the same size $k$.
That is, $k_i = k$ for all $i = 1, 2, \ldots , \alpha$.
We first investigate this approach. 
Let $\mu_{A_k}$ be the maximum number 
such that there are $\mu_{A_k}$ initial configurations 
$C_0^1, C_0^2, \ldots , C_0^{\mu_{A_k}}$ with size $k$ 
satisfying $\tau(C_0^i) \not= \tau(C_0^j)$ for all $1 \leq i < j \leq \mu_{A_k}$.
Then we are interested in upper and lower bounds on 
$\mu_{A_k}$ and $\mu_k$,
where $\mu_k = \max_{A_k} \mu_{A_k}$ denotes 
the maximum amount of information
that the swarm of $k$ robots can carry in external memory 
under this information transmission scheme.

However, assigning all symbols to initial configurations 
with the same size is not always a good approach from the view of transmission delay,
since 
the value of $k_0$ that satisfies  $\mu_{k_0} \geq \alpha$ can be large 
and, as we will show later, the transmission delay is roughly greater than 
$dist(u_S,u_R) k_0$ 
where $dist(u_S,u_R)$ is the distance between $u_S$ and $u_R$
(i.e., the number of edges in the shortest path connecting them).
Let $k_i$ be the size of $C_0^i$ assigned to $s_i$ for $i = 1, 2, \ldots , \alpha$.
Then we are interested in upper and lower bounds on the average size
$K^* = \sum_{i=1}^{\alpha} p_i k_i$, 
since there is a chance that $K^* < k_0$ holds for some $\cal A$.
The reduction of $K^*$ contributes to the reduction of energy,
as well as the reduction of transmission delay, 
whose lower bound is roughly estimated by $dist(u_S,u_R) K^*$.
However, its upper bound still heavily depends on $\cal A$. 

We mainly investigate information transmission on \emph{8-grids} 
$G_8(m,n) = ({\mathbb Z}_{m,n}^2, N_8)$,
where ${\mathbb Z}_{m,n}^2 = \{(i,j) \in {\mathbb Z} \times {\mathbb Z} 
\mid 0 \leq i \leq m-1, 0 \leq j \leq n-1 \}$ and
$N_8 =\{ ((i,j),(i',j') \in {\mathbb Z}_{m,n}^2 \times {\mathbb Z}_{m,n}^2 
: (|i - i'| \leq 1) \wedge 
(|j - j'|) \leq 1) \wedge ((i,j) \not= (i',j')) \}$.

\noindent{\bf Our contributions.~}
Our contributions can be summarized as follows:
\begin{enumerate}
 \item 
We present exponential upper bounds of $\mu_k$ on general graphs and 8-grids. 
\item
Focusing on ``narrow'' channel realized by $G_8(m,2)$,
we present two algorithms that promises exponential code size 
to a fixed swarm size $k$. 
The first algorithm promises a near optimal code size $2^{k-14}$,
but has a large transmission delay.
The second algorithm promises a near optimal transmission delay,
but has a smaller code size $2^{\lfloor k/2 \rfloor}$. 
\item
We step into the second approach with these two algorithms 
and show upper bounds of the expected swarm size. 
\item
We finally analyze lower bounds of the expected swarm size 
based on the Shannon's lemma and noiseless coding theory. 
\end{enumerate}

\noindent{\bf Related works.~}
Study of information transmission by a binary (or $q$-ary) code 
is a main stream in computer science known as information theory (or coding theory).
It is a fully developed research area,
and there are many standard textbooks (e.g., \cite{CT06}).
One can notice the similarity between the first approach and the equal-length coding, 
and the second one and the variable-length coding.
Our results are also related with Shannon's source coding theorem,
which relates the optimal average code length with the entropy of the information source.

Information transmission on a graph by a swarm of anonymous oblivious 
mobile robots
has not been investigated to the best of our knowledge 
although, including a swarm of mobile robots,
distributed systems consisting of mobile entities 
have been extensively investigated 
in the last couple of decades.  
Two survey books \cite{FPS12,FPS19} include their models and 
algorithms to solve typical problems
such as gathering and convergence \cite{CFPS12,PPV15,SY99}, 
pattern formation \cite{DPV10,FPSW08,FYOKY15,SY99,YS10,YUKY16,YY14}, 
scattering and covering \cite{CDPIM10,DP07,TKGT18,IKY14}, 
flocking and marching \cite{AFSY08,CP07,GP04,YXCD08}, 
and searching and exploration \cite{BMPT11,DLPRT12,DYKY18,FIPS13}. 

A main theme in these studies is to solve the problems without 
using memory,
partly motivated by a challenge to prejudice 
that those systems cannot solve non-trivial problems 
and partly because memoryless algorithms are usually strong
against transient failures \cite{D74}.
For example, exploration of a finite square grid by a metamorphic robot of size 5 is proposed \cite{DYKY18}.
The metamorphic robot consists of $5$ anonymous oblivious mobile robots
(called \emph{modules} in the literature) 
and remembers the current search direction 
(i.e., right, left, up or down) 
in its configuration 
during the exploration.
In the existing pattern formation algorithms for anonymous oblivious 
mobile robots, 
in order to control the move order among the anonymous robots, 
they agree on 
a common coordinate system and remember 
it in external memory \cite{DPV10,FPSW08,FYOKY15,SY99,YS10,YUKY16}. 
In each of the oblivious algorithms one can find 
a trick to maintain external memory.

\noindent{\bf Organization of the paper.~} 
We define the swarm of anonymous oblivious robots on a graph 
and formalize the noiseless communication channel realized by 
the swarm with an algorithm in Section~\ref{sec:preliminary}. 
Section~\ref{sec:eqsizecode} first presents 
an upper bound of $\mu_{G_8(m,n),k}$. 
Then, we present two algorithms for information transmission 
by a swarm of fixed size. 
We consider the second approach in Section~\ref{sec:varsizecode} 
with the two algorithms and present 
lower and upper bounds of the expected swarm size. 
We conclude this paper with Section~\ref{sec:concl} 
which also includes 
future directions and open problems.

\section{Preliminaries} 
\label{sec:preliminary}

\subsection{Swarm of anonymous oblivious mobile robots}

Let $G = (V, E)$ be a simple connected undirected graph.
For $u \in V$, $N_G(u)$ denotes the set of neighbors of $u$ in $G$, and 
for a subset $U \subseteq V$, 
$G[U]$ denotes the subgraph of $G$ induced by $U$. 
For two vertices $u,v \in V$, 
$dist_G(u,v)$ denotes the distance between $u$ and $v$. 
We define two infinite regular graphs. 
Let ${\mathbb Z}$ be the set of integers. 
The \emph{$4$-grid} is an infinite graph $G_4 = ({\mathbb Z}^2, N_4)$, 
where 
$N_4 = \{((i,j), (i', j')) \in {\mathbb Z}^2 \times {\mathbb Z}^2 \mid 
((i=i') \wedge (|j-j'|=1)) \vee ((|i-i'|=1) \wedge (j=j'))\}$ 
and the \emph{$8$-grid} is an infinite graph $G_8 = ({\mathbb Z}^2, N_8)$, 
where 
$N_8 = \{((i,j),(i'.j')) \in {\mathbb Z}^2 \times {\mathbb Z}^2 \mid 
(|i-i'|\leq 1) \wedge (|j-j'|\leq 1) \wedge ((i,j) \not= (i',j'))\}$. 

We consider a swarm $R = \{r_1, r_2, \ldots, r_k\}$ 
of anonymous oblivious mobile robots on $G$. 
We use the indices just for description since a robot is anonymous. 
A robot is \emph{oblivious}, 
which means that it does not have memory to remember information obtained in the past.
A vertex of $G$ can accommodate at most one robot at each time step, 
and the set of vertices $C \subseteq V$ that accommodate robots 
is called the \emph{configuration} of $R$. 
We say that a configuration $C$ is \emph{connected} if $G[C]$ is connected. 
A robot $r$ at $u \in V$ can move to a vertex $v \in N_G(u) \setminus C$ through edge $(u,v)$
if the new configuration $C(u;v) = (C \setminus \{u\}) \cup \{v\}$ 
obtained by this movement of $r$ is connected.
In what follows, 
we always assume that a configuration means a connected one.
By ${\cal C}$ we denote the set of (connected) configurations $C$.

Each robot $r \in R$ knows $G$ and is aware of the vertex $u$ it resides.
It repeats a Look-Compute-Move cycle once initialized.
In each Look-Compute-Move cycle,
it first observes the positions of other robots on $G$ in Look phase. 
The visibility range is finite,
but is large enough to observe all the robots,
so that $r_i$ recognizes the current configuration $C$. 
In Compute phase, given $C$ and $u$ as inputs,
a common deterministic algorithm $A$ (on $r$) computes its next position $v$
in such a way that $C(u;v)$ is connected,
where $v = u$ is possible and it means to stay at $u$. 
In this paper, 
we assume that $A$ outputs $v = u$ on all robots $r$ but one.\footnote{
Since all robots know $G$ and $C$, 
$A$ on all robots can choose the same vertex $u \in C$ 
to move the robot at $u$ to a different vertex $v \in N_G(u) \setminus C$.
Algorithm $A$ may output $v = u$ on all robots $r$.}
Finally, it moves to $v$ through $(u,v)$ in Move phase. 
Observe that an oblivious robot which does not have memory
(except the input buffer of $A$ for $C$ and $u$, and the work space for $A$) can execute $A$.

We consider discrete time $0, 1, 2, \cdots$.
Given an algorithm $A$ and an initial configuration $C_0 \in {\cal C}_k$,
all robots are initialized at time $0$,
and synchronously execute a Look-Compute-Move cycle in each time step.
Then $A$ chooses a single robot $r$ at a vertex $u \in C_0$ and moves 
it to one of its neighbors $v$ to yield a configuration $C_1 = C_0(u;v)$ at time 1.
The swarm of robots repeats this process and yields configurations $C_2, C_3, \ldots$.
We call the evolution of configurations, i.e., 
${\mathcal B} = C_0, C_1, C_2, \ldots$, 
the \emph{behavior} of algorithm $A$ 
from an initial configuration $C_0$. 
Since $A$ is deterministic,
${\mathcal B}$ from $C_0$ under $A$ is uniquely determined.

\subsection{Graph as a noiseless communication channel} 

We investigate information transmission from 
the \emph{sender} $u_S \in V$ to the \emph{receiver} $u_R \in V$ 
by a swarm of robots, i.e., we regard $G$ as a communication channel. 
The sender $u_S$ is a memoryless information source that generates 
a symbol in $S = \{s_1, s_2, \ldots, s_{\alpha}\}$,
where $s_i$ is generated with probability $\Prob(s_i) = p_i$ 
for each $i=1,2, \ldots, \alpha$. 
Thus $\sum_{i=1}^{\alpha} p_i = 1$. 
Let ${\cal C}_I$ be the set of configurations that contain $u_S$ and 
${\cal C}_T$ be the set of configurations that contains $u_R$. 
With each symbol $s_i$, 
we associate a configuration $\gamma_I(s_i) \in {\cal C}_I$ and 
a set of configurations $\gamma_T(s_i) \subseteq {\cal C}_T$, 
and regard the pair $(\gamma_I(s_i), \gamma_T(s_i))$ as the \emph{codeword} of $s_i$. 
Here $\gamma_I(s_i) \neq \gamma_I(s_j)$ 
and $\gamma_T(s_i) \cap \gamma_T(s_j) = \emptyset$, for any $i \neq j$. 
To send $s_i$ from $u_S$, as explained in Section 1,
$R$ is initialized with $\gamma_I(s_i)$ at $u_S$.
When a configuration $C_d$ of $R$ in $\gamma_T(s_i)$ is reached,
$u_R \in C_d$ receives $s_i$.
We call the pair $\gamma=\langle \gamma_I, \gamma_T \rangle$ \emph{code} of $S$,
and $|\gamma|$, which is the size of $\gamma$, is the size $\alpha$ of $S$. 

Now we consider a coding scheme with an algorithm $A$. 
Let ${\mathcal B} = C_0, C_1, \ldots$ be a behavior of $R$ under $A$. 
We define a function $\tau_A: {\cal C}_I \to {\cal C}_T \cup \{\bot\}$ as follows:
If ${\mathcal B}$ eventually reaches a configuration $C_d \in {\cal C}_T$ for the first time, 
then $\tau_A(C_0) = C_d$. 
We call $d$ the \emph{transmission delay}. 
Otherwise, if ${\mathcal B}$ does not reach a configuration in ${\mathcal C}_T$ forever,
then $\tau_A(C_0) = \bot$. 
We say algorithm $A$ \emph{realizes} code $\gamma$ 
if $\tau_A(\gamma_I(s_i)) \in \gamma_T(s_i)$ for each $i=1,2, \ldots, |\gamma|$. 
Let $\mu_A = |\{\tau_A(C) \mid C \in {\mathcal C}_I\} \setminus \{\bot\}|$.

\begin{theorem}
 \label{obs:1}
Let $A$ be any algorithm. 
\begin{enumerate} 
 \item There is a code $\gamma$ with $|\gamma| = \mu_A$ which is realizable by $A$. 
 \item There is no code $\gamma$ with $|\gamma| = \mu_A+1$ which is realizable by $A$. 
\end{enumerate}
\end{theorem}

\begin{proof}
As for $1$, we construct a code 
$\gamma^A = \langle \gamma_I^A, \gamma_T^A \rangle$ 
with $|\gamma^A| = \mu_A$ realizable by $A$ as follows: 
Let $S_{\mu_A} = \{ s_1, s_2, \ldots, s_{\mu_A}\}$.
By definition,
there are $\mu_A$ configurations $C_0^1, C_0^2, \ldots , C_0^{\mu_A}$ in ${\mathcal C}_I$
such that $\tau_A(C_0^i) \not= \bot$ for all $i = 1, 2, \ldots , \mu_A$
and $\tau_A(C_0^i) \not= \tau_A(C_0^j)$ for all $1 \leq i < j \leq \mu_A$.
We then define $\gamma_I^A(s_i) = C_0^i$ 
and $\gamma_T^A(s_i) = \{\tau_A(C_0^i)\}$ for $i = 1, 2, \ldots , \mu_A$.
It is easy to observe that $\gamma^A$ is indeed a code of size $\mu_A$. 

As for $2$, suppose that there exists a code 
$\gamma = \langle \gamma_I, \gamma_T \rangle$ with $|\gamma|=\mu_A +1$ 
to derive a contradiction. 
Since $|\gamma|=\mu_A +1 > \mu_A$, there are two symbols $s_i$ and 
$s_j(\neq s_i)$ such that $\tau_A(\gamma_I(s_i)) = \tau_A (\gamma_I(s_j))$, 
a contradiction. 
\end{proof}

The proof of Theorem~\ref{obs:1} gives a construction method 
to construct code $\gamma^A$ from algorithm $A$. 
The transmission delay of $\gamma^A$ depends on the choice of 
$\gamma_I^A(s_i) \in {\mathcal C}_I$ for each $s_i \in S_{\mu_A}$;
any $C$ such that $\tau_A(C) = \tau_A(\gamma_I^A(s_i))$ can be chosen instead.
Let $d_A(C)$ be the transmission delay of the behavior of $R$ starting from $C$ under $A$. 
To reduce the transmission delay of $\gamma^A$, 
it is natural to choose an initial configuration $C$ that minimizes $d_A(C)$, 
and we thus assume that such a $C$ is chosen in the construction of $\gamma^A$. 
Then the transmission delay $d_{\gamma^A}$ of $\gamma^A$ 
is ${\displaystyle \max_{s_i \in S_{\mu_A}} d_A(\gamma_I^A(s_i)) }$. 
We are interested in the maximization of $\mu_A$ and the minimization of $d_A$.
Let ${\displaystyle \mu_{N,k} = \max_A \mu_A }$ and ${\displaystyle d_{N,k} = \min_A d_A }$,
where ${\displaystyle d_A = \max_{C: \tau_A(C) \neq \bot} d_A(C) }$ and $N=\langle G, u_S, u_R \rangle$.

\subsection{General graphs}

Since $\mu_{N,1} = 1$, 
we assume $k \geq 2$ in what follows. 
We also assume $dist_G(u_S,u_R) \geq k$, i.e., 
${\mathcal C}_I \cap {\mathcal C}_T = \emptyset$. 
Indeed we are interested in the case where $dist_G(u_S, u_R) \gg k$.   

\begin{observation}
\label{O201}
For any $N=\langle G, u_S, u_R \rangle$,
$\mu_{N,k} \leq |{\mathcal C}_I|$. 
\end{observation}

\begin{theorem}
\label{T201}
Let $N=\langle G, u_S, u_R \rangle$ with $dist_G(u_S, u_R) \geq k$. 
If $G$ has a cut vertex whose removal disconnects $u_S$ and $u_R$, 
then $\mu_{N,k} = 0$, that is, 
$\tau_A(C_0) = \bot$ for any algorithm $A$ and any configuration $C_0 \in {\mathcal C}_I$.
\end{theorem}

\begin{proof}
To derive a contradiction, 
we assume that there is an algorithm $A$ and a configuration $C_0 \in {\mathcal C}_T$ 
such that the behavior ${\mathcal B} = C_0, C_1, \ldots$ of $R$ under $A$ 
eventually reaches a configuration $C_d \in {\mathcal C}_T$. 
Let $v \in V$ be a cut vertex of $G$ whose removal from $G$ disconnects $u_S$ and $u_R$. 
Thus $G[V \setminus \{v\}]$ consists of more than one connected components.
Let $G_1 = (V_1 \setminus \{v\}, E_1)$ be the one that contains $u_S$.
We first assume $C_0 \subseteq V_1$.
Let $t$ be the smallest index such that $C_t \subseteq V_1$ and $C_{t+1} \not\subseteq V_1$. 
Then there is a vertex $w \in V \setminus V_1$ such that $w \in C_{t+1}$ and $(v,w) \in E$,
and the robot at $v$ at $t$ moves to $w$.
It is a contradiction since $v \not\in C_{t+1}$ and hence $C_{t+1}$ is not connected.

Next assume that $C_0 \not\subseteq V_1$.
Since $dist_G(u_S,u_R) \geq k$,
there is a time instant $t$ such that the robot $r$ at $v$ moves for the first time.
Then $C_{t+1}$ is not connected since $v \not\in C_{t+1}$,
and a contradiction is derived.
\end{proof}

\begin{theorem}
\label{T202}
Let $N = \langle G_4, u_S, u_R \rangle$ with $dist_{G_4}(u_S, u_R) \geq k$. 
Then $\mu_{N,k} = 0$, that is, 
for any algorithm $A$ and for any configuration $C_0 \in {\mathcal C}_I$, 
$\tau_A(C_0) = \bot$. 
\end{theorem}

\begin{proof}
To derive a contradiction, 
we assume that there is an algorithm $A$ and a configuration $C_0 \in {\mathcal C}_I$ 
such that the behavior $B= C_0, C_1, \ldots $ of $R$ under $A$ 
eventually reaches a configuration $C_d \in {\mathcal C}_T$. 
Since $dist_G(u_S, u_R) \geq k$, 
$d \geq 1$ and $|i_R - i_S| + |j_R - j_S| \geq k$, 
where $u_S = (i_S, j_S)$ and $u_R = (i_R, j_R)$. 
We assume without loss of generality that $i_R > i_S$ and $j_R \geq j_S$. 

For configuration $C_t$, let $x_{\max}(t) = \max \{i \mid (i,j) \in C_t\}$ and 
$y_{\max} = \max \{j \mid (i,j) \in C_t \}$. 
If $x_{\max}(0) \geq i_R$, then $j_S \leq y_{\max} < j_R$. 
Thus either $i_S \leq x_{\max}(0) < i_R$ or $j_S \leq y_{\max}(0) < j_R$ holds. 
We assume without loss of generality that 
$i_S \leq x_{\max} < i_R$. 

Let $t$ be the smallest index such that $x_{\max}(0) + 1 = x_{\max}(t)$ holds. 
That is, $x_{\max}(t-1) = x_{\max}(0)$. 
Suppose that the robot $r$ at vertex $u \in C_{t-1}$ moves. 
Then $C_t$ is not connected, 
since $r$ is the only robot with $x$-coordinate $x_{max}(0)+1$ in $C_t$ 
and $u \not\in C_t$.
\end{proof}

\begin{theorem}
\label{theorem:computable}
Let $G$ be a finite graph and $N = \langle G, u_S, u_R \rangle$. 
Then, $\mu_{N,k}$ 
and $d_{N,k}$
are computable. 
\end{theorem}

\begin{proof}
We consider the transition diagram (i.e., directed graph)
$X = ({\mathcal C} \cup \{ v_S, v_R \},F)$ of the swarm on $G$,
where $(C,C') \in F$ if and only if 1) there are $u \in C$ and $v \in C'$ such that $C' = C(u;v)$,
2) $C = v_S$ and $C' \in {\mathcal C}_I$, 
or 3) $C \in {\mathcal C}_T$ and $C' = v_R$.
Then $\mu_{N,k}$ is the size of the maxflow of $X$ from $v_S$ to $v_R$. 

As for $d_{N,k}$,
there are only a finite number of different sets $\Pi$ of $\mu_{N,k}$ disjoint paths 
connecting $v_S$ and $v_R$.
By comparing the length of the longest path in each $\Pi$,
we can calculate $d_{N,k}$.
\end{proof}

By Theorem~\ref{theorem:computable}, 
we can design an optimal algorithm $A$ to achieve $\mu_{N,k}$
such that $A$ makes $R$ follow one of the $\mu_{N,k}$ disjoint paths connecting $v_S$ and $v_R$.
Although an optimal $A$ exists, for a general $N$,
it is in generally hard to explicitly describe $A$ and to estimate $\mu_{N,k}$ or $d_{N,k}$.
However, we have the following rather obvious upper bound of $d_{N,k}$ 
when the distance between $u_S$ and $u_R$ is large.

\begin{theorem}
\label{T204}
Let $G$ be a finite graph and $N = \langle G, u_S, u_R \rangle$,
where $dist_G(u_S, u_R) \geq 2k-1$ and
$C_0 \cap C_d = \emptyset$ for any $C_0 \in {\cal C}_I$ and $C_d \in {\cal C}_T$.
Then $d_{N,k} \geq k(dist_G(u_S, u_R) - 2(k-1))$.
\end{theorem}

\begin{proof}
For any algorithm $A$,
let ${\mathcal B} = C_0, C_1, C_2, \ldots$ 
be the behavior from any initial configuration $C_0 \in {\mathcal C}_I$ under $A$.
We assume that ${\mathcal B}$ eventually reaches a configuration $C_d \in {\mathcal C}_T$.
For each robot $r_i \in R$, let $x_i$ and $y_i$ be the vertices that it resides
at time 0 and $d$, respectively.
Thus $x_i \in C_0$ and $y_i \in C_d$.
Obviously $dist_G(u_S, x_i) \leq k-1$ and $dist_G(u_R, y_i) \leq k-1$.
Hence each robot moves at least $dist_G(u_S, u_R) - 2(k-1)$ until time $d$,
which implies that $d_A(C_0) \geq k (dist_G(u_S, u_R) - 2(k-1))$. 
\end{proof}

The general policy to construct a code with a small transmission delay
is to use a swarm with a small number of robots.

\section{Finite 8-grids}
\label{sec:eqsizecode}

We investigate $8$-grids in the rest of this paper.
Let ${\mathbb Z}_{m,n}^2 = \{(i,j) \in {\mathbb Z} \times {\mathbb Z} \mid 
0 \leq i \leq m-1, 0 \leq j \leq n-1\}$ 
and $G_8(m,n) = G_8[{\mathbb Z}_{m,n}^2]$. 
Clearly, $\mu_{\langle G_8(m,1), u_S, u_R \rangle, k} = 0$ and
$d_{\langle G_8(m,1), u_S, u_R \rangle, k} = \infty$ for any $m > k \geq 2$. 
We illustrate a swarm in a 2D square grid instead of $G_8(m,n)$. 
Each cell of the square grid is associated with the underlying 
$x$-$y$ coordinate system and each vertex $(i,j) \in G_8(m, n)$ corresponds to cell $(i,j)$. 
Hence $N_{G_8(m,n)}((i,j))$ corresponds to the eight adjacent cells, 
i.e., 
$(i+1,j), (i+1,j+1),(i,j+1),(i-1,j+1),(i-1, j),(i-1,j-1),(i,j-1),(i+1,j-1)$. 
We call the cells $\{(i, j) \mid i=0,1,2,\ldots, m-1\}$ the \emph{$j$th row} and 
the cells $\{(i, j) \mid j=0,1,2,\ldots, n-1\}$ the \emph{$i$th column}. 
For the sake of simplicity, 
we assume $u_S = (0, 0)$ and $u_R = (m-1, 0)$.

\subsection{Upper bound of $\mu_{G_8(m,n),k}$}

We derive an upper bound of $\mu_{(G_8(m,n),k)}$ by estimating the size $|{\mathcal C}_I|$
of initial configurations by Observation~\ref{O201}.
Since $u_S$ and $u_R$ are fixed when $G_8(m,n)$ is given,
instead of $N = \langle G_8(m,n), u_S, u_R \rangle$, 
we omit $u_S$ and $u_R$ to denote $\mu_{G_8(m,n),k}$.

\begin{lemma}
\label{L301}
For any $u_S$ and $u_R$,
the number of initial configurations of $G_8$ is at most $2^{6(k-1)}$. 
\end{lemma}

\begin{proof}
Let $C \in {\mathcal C}_I$ be any initial configuration. 
Consider the depth first traversal of an arbitrary 
spanning tree of $G[C]$ rooted at $u_S$. 
The track is represented by a sequence of the eight types of movements, 
$(+1, 0)$, $(+1, +1)$, $(0, +1)$, $(-1, +1)$, 
$(-1, 0)$, $(-1, -1)$, $(0, -1)$, and $(+1, -1)$. 
Thus any configuration $C \in {\mathcal C}_I$ is 
represented by a sequence of movements whose length is $2(k-1)$. 
Hence $|{\mathcal C}_I| < 8^{2(k-1)} = 2^{6(k-1)}$.
\end{proof}

\begin{corollary}
\label{C301}
For any $m, n \geq 1$, $\mu_{G_8(m,n),k} \leq 2^{6(k-1)}$.
\end{corollary}

\begin{proof}
The number of initial configurations in  $G_8(m,n)$ is obviously
not greater than that of $N = \langle G_8, u_S, u_R \rangle$.
\end{proof}

When $n = 2$, we can obtain a better bound of $\mu_{G_8(m,2),k}$.
A trivial upper bound of $\mu_{G_8(m,2),k}$ is $3^k$,
since $|{\mathcal C}_I| \leq 3^k$
and each column of any initial configuration in ${\mathcal C}_I$
must contain at least one robot to keep the connectivity. 

\begin{lemma}
\label{L302}
If $m \geq k$,
$\mu_{G_8(m,2)} < (1 + \sqrt{2})^k$.
\end{lemma}

\begin{proof}
Recall that $u_S = (0,0)$ and $u_R = (m-1,0)$.
Let ${\mathcal C}_I$ be the set of initial configurations of 
the swarm of $k$ robots on $G_8(m,2)$.
There are two types of configurations in ${\mathcal C}_I$:
\begin{enumerate}
 \item 
The set of configurations $C \in {\mathcal C}_I$ 
such that $(0,1) \not\in C$.
 \item 
The set of configurations $C \in {\mathcal C}_I$ 
such that $(0,1) \in C$.
\end{enumerate}
Note that $(0,0) \in C$ for any $C \in {\mathcal C}_I$ by definition.
Let $n_k$ be the number of configurations such that the $0$th column contains 
at least one robot. 
Such a configuration may not contain $(0,0)$.
We have the following recurrence formula:
\begin{equation*}
n_{k} = 2 n_{k-1} + n_{k-2},
\end{equation*}
with $n_{1} = 2$ and $n_{2} = 5$. 
We can calculate the general term $n_k$ using a standard method.
Since the two roots of equation $x^2 - 2x - 1 = 0$ are $1 \pm \sqrt{2}$,
we have two equations
\begin{eqnarray*}
n_{k} - (1-\sqrt{2})n_{k-1} &=& (1+\sqrt{2})^{k-2}(3 + 2\sqrt{2}), {\rm~~and} \\ 
n_{k} - (1+\sqrt{2})n_{k-1} &=& (1-\sqrt{2})^{k-2}(3 - 2\sqrt{2}). 
\end{eqnarray*}
Thus
\begin{equation*}
n_k = \frac{3+2\sqrt{2}}{2\sqrt{2}} (1+\sqrt{2})^{k-1} 
- \frac{3-2\sqrt{2}}{2\sqrt{2}} (1-\sqrt{2})^{k-1}.
\end{equation*}
Since $|{\mathcal C}_I| = n_{k-1} + n_{k-2}$,
\begin{eqnarray*}
|{\mathcal C}_I| &=& 
\frac{3+2\sqrt{2}}{2\sqrt{2}} \left((1+\sqrt{2})^{k-2} + (1+\sqrt{2})^{k-3} \right) 
- \frac{3-2\sqrt{2}}{2\sqrt{2}} \left((1-\sqrt{2})^{k-2} + (1-\sqrt{2})^{k-3} \right) \\ 
&=& \sqrt{2} \left(1 + \frac{3}{2\sqrt{2}} \right)(1+\sqrt{2})^{k-2}  
+ \sqrt{2} \left(1 + \frac{3}{2\sqrt{2}} \right) (1-\sqrt{2})^{k-2} \\ 
&<& \sqrt{2}(1+\sqrt{2})^{k-1} + \sqrt{2}(1+\sqrt{2}) \\ 
&<& (1+\sqrt{2})^{k-1}(1+\sqrt{2}) \\ 
&=& (1+\sqrt{2})^{k}. \\
\end{eqnarray*}

\end{proof}

\subsection{Lower bound of $\mu_{G_8,k}$}

We present an algorithm $ALG_1$ on $G_8(m,2)$ 
by which we can construct a code $\gamma$ with size $|\gamma| = 2^{k-14}$. 
Thus $2^{k-14} \leq \mu_{G_8(m,2),k} \leq \mu_{G_8(m,n),k} \leq \mu_{\langle G_8, u_S, u_R \rangle,k}$ 
hold for any $n \geq 2$.
Referring to Lemma~\ref{L302}, $ALG_1$ produces a code with a near optimal code size $2^{k-14}$.
However, the transmission delay of the code 
is not so good.
We hence propose another algorithm $ALG_2$ by which we can produce a code with 
an optimal transmission delay, in the next subsection,
although its code size is small and $2^{\lfloor k/2 \rfloor}$.

Algorithm $ALG_1$ divides the swarm into a codeword and a controller.
Roughly, the codeword is external memory to remember the codeword of the symbol transmitted,
and the controller, which is placed following the codeword,
is used to keep the external memory stable.
The controller consists of $14$ robots,
which is further divided into a buffer consisting of four robots, 
a temporal memory consisting of five robots, a signal consisting of one robot, 
and a counter consisting of four robots. 

In the canonical state,
the configuration $C$ contains exactly one robot at each column as illustrated in Figure \ref{figure:G8-m-2-control}. 
Algorithm $ALG_1$ shifts $C$ right by {\em two} (not one) columns.
Suppose that the leftmost cell of $C$ is in the column $i_s$.
Then the rightmost cell is in the column ($i_s + k - 1$).
The leftmost $(k - 14)$ robots in $C$ represents the symbol it is transmitting.
We can thus transmit 
one of $2^{k-14}$ symbols or, in other words,
this part can be regarded as memory of $(k-14)$ bits,
by associating the robot in row 0 (resp. row 1) with 0 (resp. 1).
The configuration in Figure \ref{figure:G8-m-2-control} 
is thus transmitting a binary sequence 011010100110. 

We regard the swarm as a sequence of arrays whose element can take $0$ or $1$;
the codeword array $code[0..k-15]$, 
followed by the buffer array $bf[0..4]$, the temporal memory array $cp[0..4]$, 
the signal array $sign[0]$, and the counter array $cnt[0..3]$ (Figure \ref{figure:G8-m-2-control}). 

\begin{figure}[tbp] 
 \centering
 \includegraphics[width=8cm]{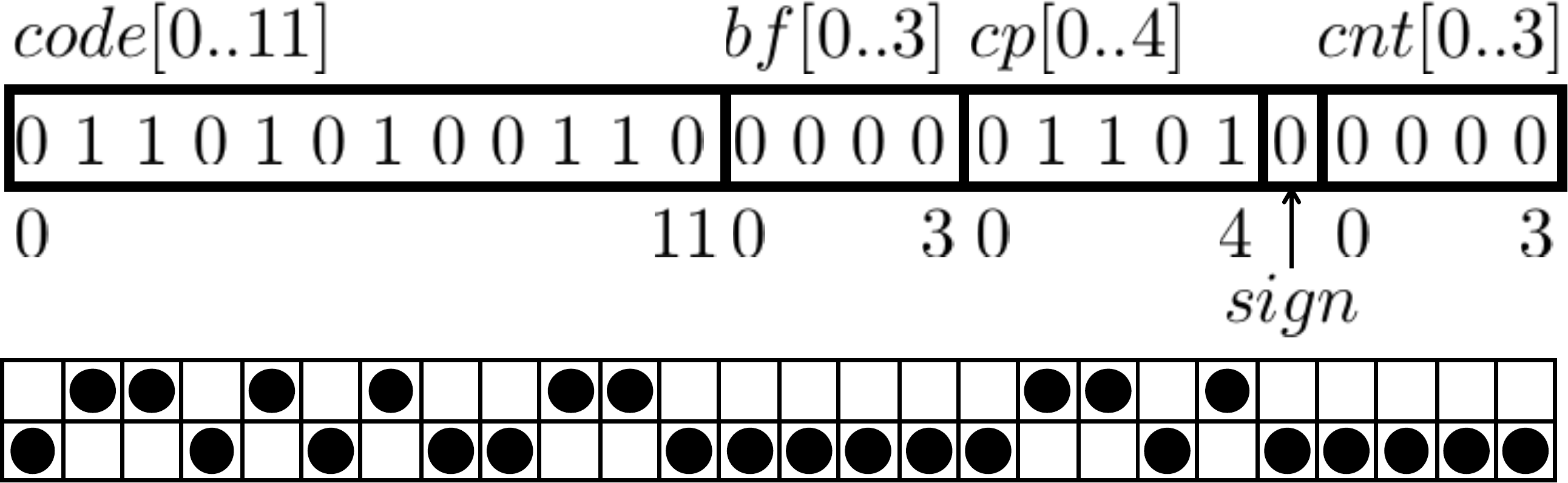}
 \caption{Illustration of Algorithm $ALG_1$ for $G_8(m,2)$. 
It illustrates a configuration in the canonical state.
The array top is an array representation of the configuration bottom.} 
 \label{figure:G8-m-2-control} 
\end{figure}

\begin{figure}[tbp] 
 \subfloat[Configuration carrying $011010100110$ after setting cp. ]{
 \includegraphics[height=0.5cm]{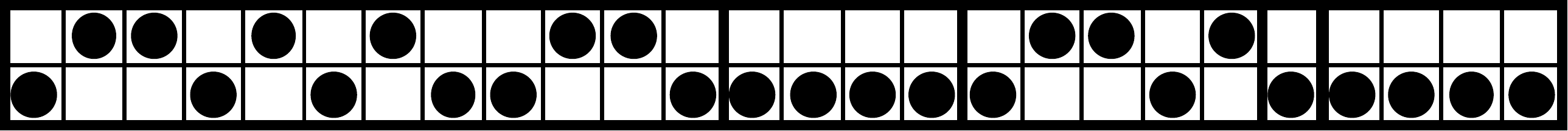}
 } \hspace{3mm}
 \subfloat[Initial wave is set. ]
 {\includegraphics[height=0.5cm]{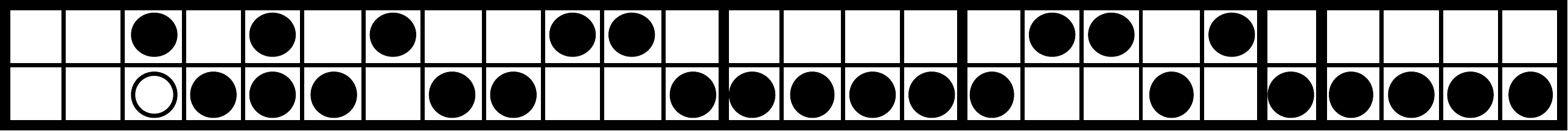}
 }\\ 
 \subfloat[First bit $b_1 = 0$ is created by the move of correct robot.]{
 \includegraphics[height=0.5cm]{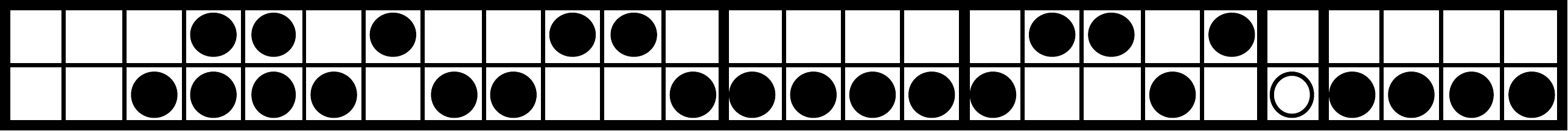}
 }\hspace{3mm}
 \subfloat[Signal is set to $1$.]{
 \includegraphics[height=0.5cm]{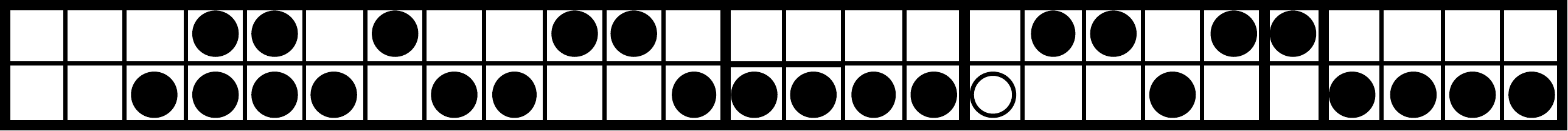}
 }\\ 
 \subfloat[The first bit of the copy is updated.]{
 \includegraphics[height=0.5cm]{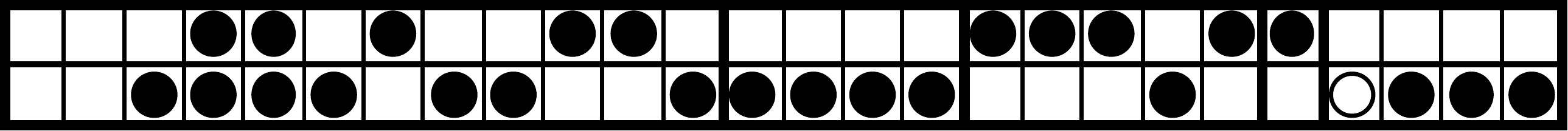}
 }\hspace{3mm}
 \subfloat[Counter is incremented. ]{
 \includegraphics[height=0.5cm]{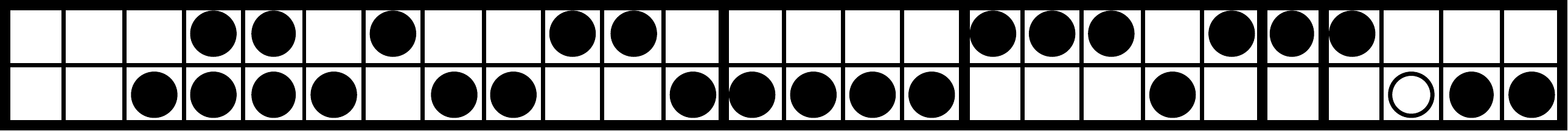}
 }\\ 
 \subfloat[Counter is incremented. ]{
 \includegraphics[height=0.5cm]{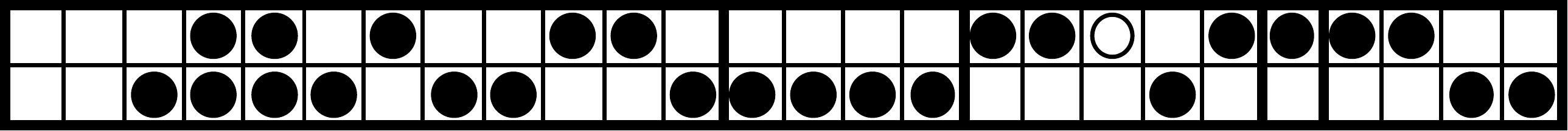}
 }\hspace{3mm}
 \subfloat[The third bit of the copy is updated. ]{
 \includegraphics[height=0.5cm]{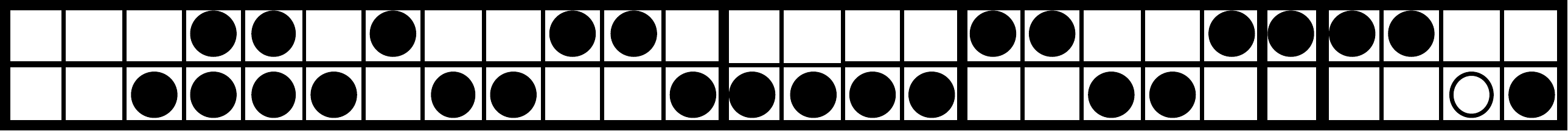}
 }\\ 
 \subfloat[Counter is incremented.]{
 \includegraphics[height=0.5cm]{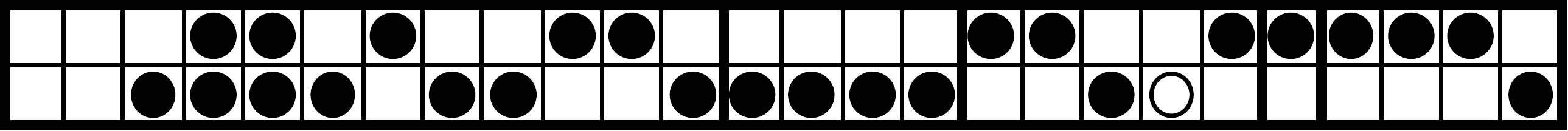}
 }\hspace{3mm}
 \subfloat[The fourth bit of the counter is updated.]{
 \includegraphics[height=0.5cm]{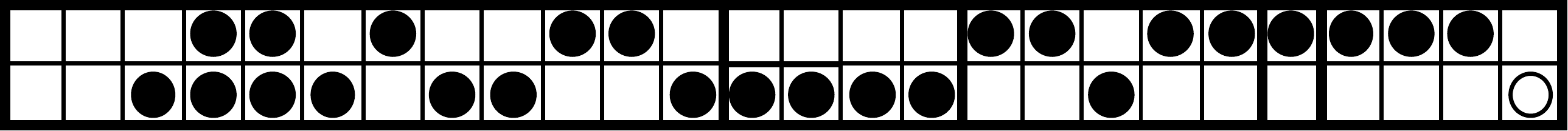}
 }\\ 
 \subfloat[Counter is incremented. ]{
 \includegraphics[height=0.5cm]{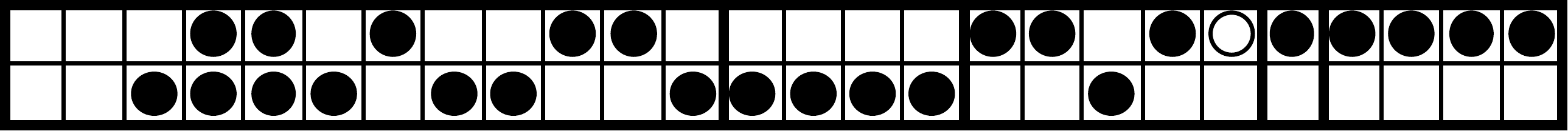}
 }\hspace{3mm}
 \subfloat[The last bit of the copy is updated. ]{
 \includegraphics[height=0.5cm]{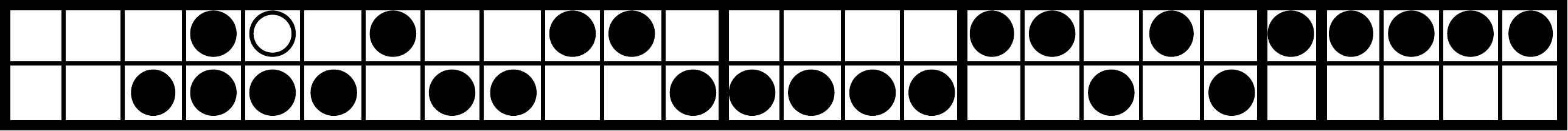}
 }\\ 
 \subfloat[Head moves. ]{
 \includegraphics[height=0.5cm]{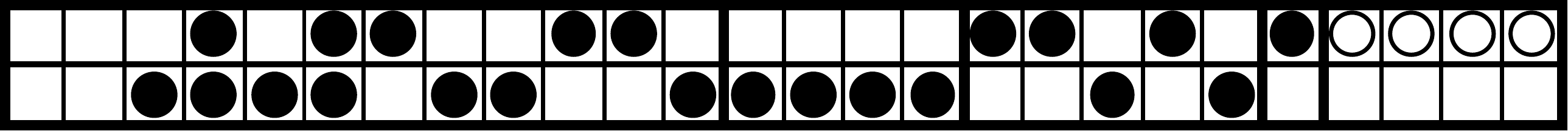}
 }\hspace{3mm}
 \subfloat[Counter is reset. ]{
 \includegraphics[height=0.5cm]{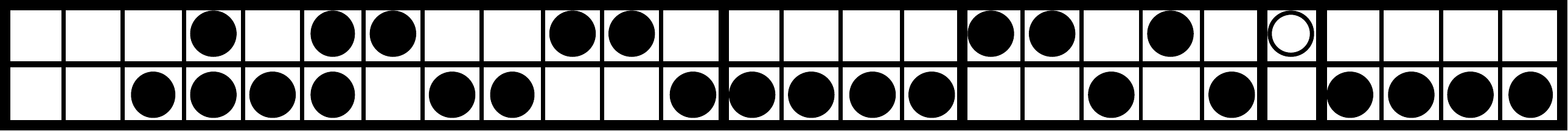}
 }\\     
 \subfloat[Signal is reset.]{
 \includegraphics[height=0.5cm]{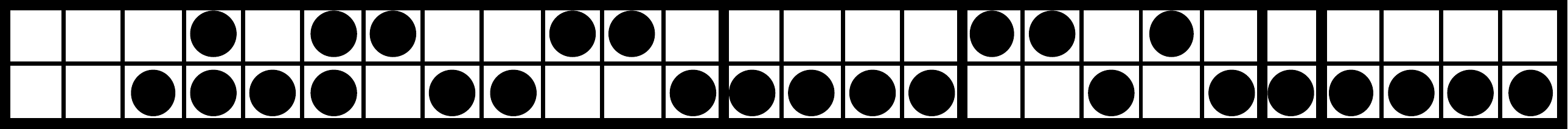}
 }
\caption{Illustrations of a right shift by 2 columns realized under Algorithm $ALG_1$. 
Only its core part is illustrated.
A robot represented by a black circle is stationary,
while a one represented by a white circle moves to yield a new configuration. 
The controller and the wave are synchronized by the values of $(head-tail)$ and $sign$. }
 \label{figure:shift}
\end{figure}

\begin{figure}[tbp] 
 \subfloat[][Final bit of the code $011010100110$ is fixed.]{ 
    \includegraphics[height=0.5cm]{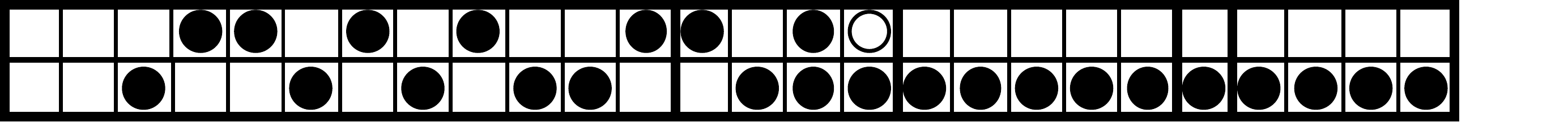}
 }\hspace{3mm} 
 \subfloat[][The head of the wave moved.]{ 
 \includegraphics[height=0.5cm]{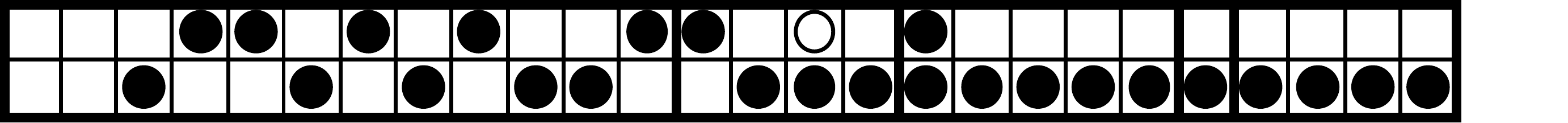}
 }\\ 
 \subfloat[][The tail of the wave follows the head.]{ 
 \includegraphics[height=0.5cm]{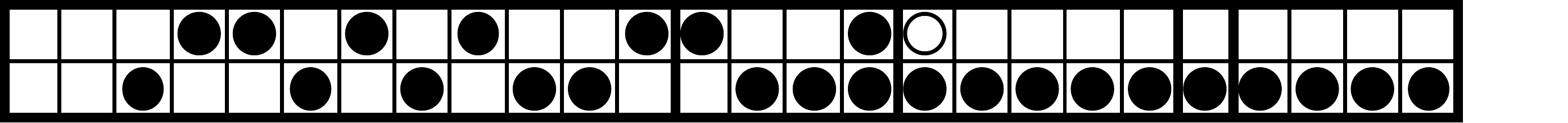}
 }\hspace{3mm}
 \subfloat[][The head and the tail repeat these movements.]{ 
 \includegraphics[height=0.5cm]{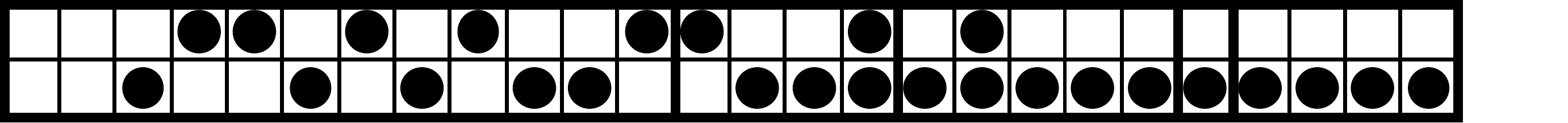}
 }\\ 
 \subfloat[][The head and the tail reach the right end of the swarm.]{ 
 \includegraphics[height=0.5cm]{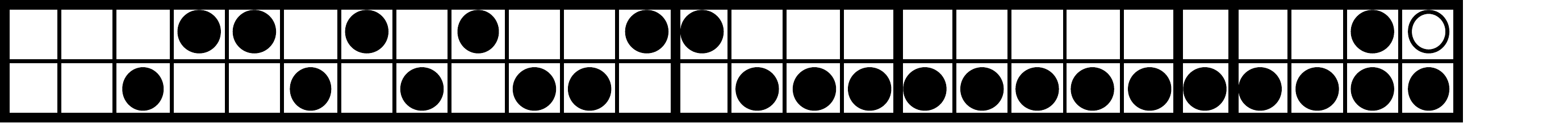}
 }\hspace{3mm}
 \subfloat[][The head of the wave moves to the $0$th row.]{ 
 \includegraphics[height=0.5cm]{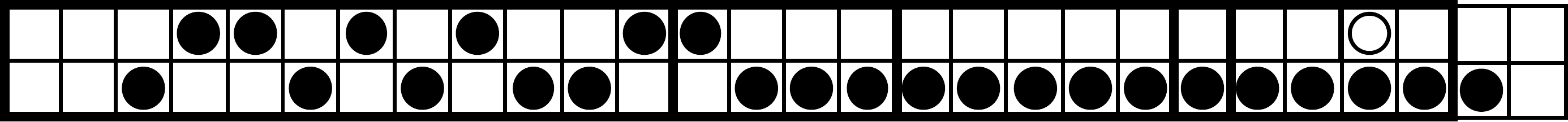}
 }\\ 
 \subfloat[][The tail of the wave moves till the end of the swarm.]{
 \includegraphics[height=0.5cm]{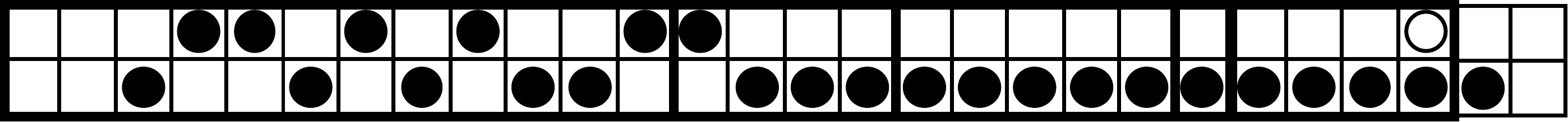}
 }\hspace{3mm}
 \subfloat[][The tail of the wave moves till the end of the swarm.]{
 \includegraphics[height=0.5cm]{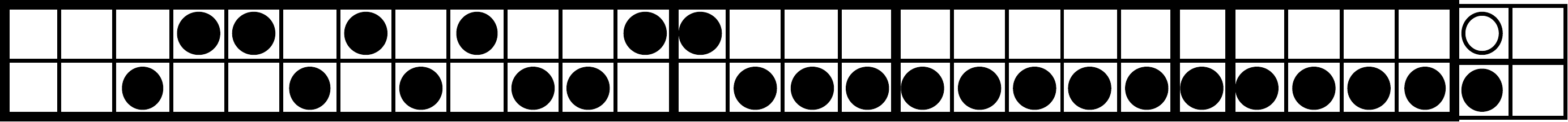}
 }\\ 
 \subfloat[][The swarm retains the canonical shape shifted by two columns.]{ 
 \includegraphics[height=0.5cm]{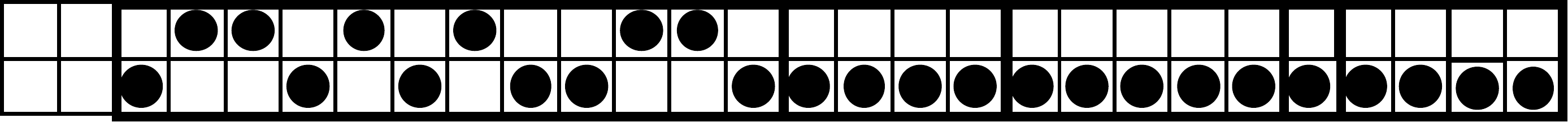}
 }
 \caption{Illustrations of the final procedure of a shift in $ALG_1$. After the final two bits are fixed, the wave go through the controller and the swarm retains the canonical state.}
\label{figure:shift-final}
\end{figure}

Starting from a configuration in the canonical state,
$ALG_1$ changes the bits in $code[0..k-1]$ one by one,
and shifts the swarm right by two columns 
when the configuration returns to the canonical state.
Here it is worth noting that a bit may be changed more than once.

First, $cp[0..4]$ copies $code[0..4]$ and 
two robots in $code[0..1]$ move to make a ``wave''. 
There are at most two columns containing two robots during the shift procedure.
When there are two such columns,
the right one is called the \emph{head} (column of the wave) 
and the left one the \emph{tail} (column of the wave). 
The shift procedure shifts a bit from left to right,
and the wave indicates the bits that $ALG_1$ is currently shifting.
However, when the wave is created, 
the bits maintained in the head and the tail are lost.
Array $cp$ is used to save these bits.
When the shift of the codeword finishes,the wave is in the controller,
and the two robots in the head and the tail are simply sent to the right end 
of the configuration to reset the controller to all zero.

The core of the shift procedure is to shift a bit in codeword right by two columns.
Let $head$ (resp. $tail$) denote the $x$-coordinate of the head (resp. the tail).
The shift is indeed carried out at the tail.
The part of the configuration to the left of the tail always remembers
a prefix $H = b_1 b_2 \ldots b_h$ of the bit sequence $B = b_1 b_2 \ldots b_{k-14}$ transmitting.
To extend $H$ by appending $b_{h+1}$,
$ALG_1$ moves the correct robot in the tail to right
so that the correct bit $b_{h+1}$ is created in the tail.
As a matter of fact, the tail has already moved to right, though.
The shift procedure is a repetition of the following phases: 

\begin{enumerate}
\item
The controller updates $cp[0..4]$ to $b_{h+1} b_{h+2} \ldots b_{h+5}$.
\item 
The correct robot in the tail moves to the next column to create bit $b_{h+1}$, 
by referring to $cp[0]$ to decide the correct robot.
This move of the robot shifts the tail right by one column.
\end{enumerate}

To implement the above procedure,
the update and the wave shift phases have to be synchronized.
Although the second phase finishes in one step,
the first phase needs a number of steps to complete, 
and the update phase must be finished before the shift of tail starts.
Algorithm $ALG_1$ keeps inequality $head-tail \leq 2$ during the procedure.
The end of phase 2 is indicated by $tail = head-1$,
and roughly until the following phase 1 finishes,
a robot in the head does not move and hence equation $tail = head-1$ holds.
When the update of $cp[0..4]$ in phase 1 finishes with cooperation of $cnt$,
a robot in the head 
moves to the next column.
When it moves, equation $tail = head-2$ holds,
which however is not a sign to finish 
phase 1 (and to start phase 2),
since $cnt$ must be reset before starting phase 2.
The completion of the reset is signaled by $sign$.
It is set once $tail = head-1$ holds.
It is reset, when $tail = head-2$ holds and $cnt$ is reset to all 0.

The update of $cp$ is carried out in the following way:
Since $cp$ stores $b_{h+1} b_{h+2} \ldots b_{h+5}$,
$ALG_1$ updates it to $b_{h+2} b_{h+3} \ldots b_{h+6}$.
For $i = 1, 2, 3$, in this order, 
$cp[i]$ is copied to $cp[i-1]$,
and finally it stores $b_{h+6}$ to $cp[4]$,
where $b_{h+6}$ is stored to the right of the head and hence it is possible.
The synchronization of these five updates are done with the cooperation of $cnt$.
Here $cnt[0..3]$ works as a unary counter, 
and it is incremented whenever an update finishes 
(or the update is not necessary because it has already had the correct bit value).
Once all bits in $cnt$ are set to 1, as explained, the head starts shifting.

Finally, the buffer $bf[0..3]$ is used when $ALG_1$ is going to shift
two bits $b_{k-15}$ and $b_{k-14}$ which were stored in the last two columns.
These column are now the tail and the head, 
and the bits are stored in $cp[0..1]$.
The buffer is used to store $b_{k-15}$ and $b_{k-14}$.
It is also necessary to separate $code$ and $cp$, 
so that the wave will not directly shift into $cp$ and the bits will not be lost.

In the following description of Algorithm $ALG_1$,
$cnt$ increments as $0000 \to 1000 \to 1100 \to 1110 \to 1111$,
and the counter value is $c$ when the number of $1$'s in $cnt[0..3]$ is $c$. 

\noindent{\bf Algorithm $ALG_1$}
\begin{enumerate}
\item 
When the configuration is canonical,
first store $b_{1} b_{2} \ldots b_{5}$ to $cp$,
then move the robots in the first and the second column to the third and the fifth
column to initialize the wave where $head = 3$ and $tail = 1$, i.e., $head - tail = 2$ 
holds after the moves.

\item
When there are no columns in $code$ that contain two robots,
but there are head and tail of the wave, 
move the two robots to right and append to the right end of the configuration,
which move reset the controller all 0,
and the configuration returns to the canonical state.

\item
When the wave is in $code$ and $sign=0$
(Note that $cnt = 0000$ always holds in this case.):
\begin{enumerate}
\item 
If $head-tail = 2$, referring to $cp[0]$,
the correct robot in the tail moves to an empty cell of the $(tail+1)$st column.
\item 
If $head-tail = 1$, set $sign$ to $1$. 
\end{enumerate}
\item 
When the wave is in $code$ and $sign = 1$:
\begin{enumerate}
\item 
If $cnt = c < 4$,
copy 
$cp[c+1]$ to $cp[c]$ if they are different.
Then increment $cnt$.
\item 
If $cnt = c = 4$, 
copy 
$code[head+1]$ to $cp[4]$ if they are different.
(As a result, $cp[4] = code[head+1]$ holds.)
If $head-tail = 1$,
one of the two robots in the head moves to an empty cell of column $head+1$. 
If $head-tail = 2$, first reset $cnt$ to 0,
and then reset $sign$ to 0 (when $cnt = 0$).
\end{enumerate}
\end{enumerate}

\begin{lemma}
\label{L303}
There is a code $\gamma$ based on algorithm $ALG_1$
whose size is $|\gamma| = 2^{k-14}$.
The transmission delay $d_{ALG_1}$ is at most $(10k-123.5)(m - k)$.   
\end{lemma}

\begin{proof}
The fact $|\gamma| = 2^{k-14}$ is obvious from the definition of algorithm $ALG_1$.

As for the transmission delay,
one cycle of transitions shifts a configuration in the canonical state right by two columns.
During the movement, 
the wave moves from the first two bits of the codeword to the end of the counter. 
Thus the total number of shifts of the wave is $2k$. 
For each bit of the codeword, 
each bit of the counter moves twice, the signal twice,  
and each bit of the copy part at most once. 
Hence, the total number of robot moves in the controller is $(k-14)(5+2+8) = 18k-252$. 
In the beginning of the shift, 
each bit of the copy part moves at most once. 
Hence the total number of robot moves to move the codeword by one column is at most
\begin{equation*}
\frac{2k + (18k-252) + 5}{2} = 10k-123.5.
\end{equation*}
The distance from the last bit of the counter to $u_R$ is at most $(m - k)$. 
Hence $d_{ALG_1} \leq (10k-123.5)(m-k)$.
\end{proof}

An obvious consequence of Lemma~\ref{L303} is that
$2^{k-14} \leq \mu_{G_8(m,2),k} \leq \mu_{G_8(m,n)}$ for all $n \geq 2$.
The transmission delay of the code by $ALG_1$ is roughly $10k(m-k)$ when $m (\gg k)$ is large,
which is not so fast compared with its lower bound $k(m-2(k-1))$ of Theorem~\ref{T204}.
This fact motivates the study of a faster algorithm in the next subsection.

\subsection{Algorithm $ALG_2$ for faster information transmission}

It is easy to design a code with a small transmission delay 
if we do not need to care the code size.
In this subsection, we design a code based on an algorithm $ALG_2$ 
whose code size is $2^{\lfloor k/2 \rfloor}$.
However its transmission delay is near optimal.

For simplicity, we assume that $k$ satisfies 
$\lfloor k/2 \rfloor = \ell$ 
for some positive integer $\ell \geq 2$. 
Algorithm $ALG_2$ divides the swarm into a codeword and a copy,
each of which consists of $\ell$ robots. 
In the same way as $ALG_1$, 
the codeword represents $\ell$ bits and the copy follows the codeword. 
When $k$ is odd, a single robot follows the copy,
and it does not represent any bit.

In the canonical state, 
each column contains exactly one robot as illustrated in 
Figure~\ref{fig:shift-by-1}. 
Algorithm $ALG_2$ shifts the canonical state by one column. 
Suppose that the leftmost cell of the canonical state is in the 
column $i$. 
Then the rightmost cell is in the column $(i+k-1)$. 
The left $\ell$ robots represent the symbol it is transmitting. 
We can thus transmit one of $2^{\ell}$ symbols, or 
in other words, these $\ell$ robots can be regarded as memory of $\ell$ bits 
$B =b_1 b_2 \ldots b_{\ell}$. 
We now regard the robots in the left $\ell$ columns as an array 
$code[0..\ell-1]$ and 
those in the right $\ell$ columns as an array $copy[0..\ell-1]$. 
When $k$ is odd, the last robot is put in the $0$th row of 
column $i+k-1$.
The canonical state in Figure~\ref{fig:shift-by-1} represents
a binary sequence $010$. 

Starting from a configuration where the swarm is in the canonical state, 
$ALG_2$ shifts the swarm right by one column when the configuration returns 
to the canonical state. 
In the canonical state, $ALG_2$ first shifts the codeword by 
generating a wave; 
the leftmost robot moves to the empty cell of its right column. 
From now on, we regard the column with two robots as the wave.\footnote{
The wave in $ALG_2$ consists of just one column and this is different from 
the wave in $ALG_1$, that consists of two columns for synchronization.} 
Once the wave is initialized, the wave moves to right 
by repeating the following procedure; 
when the wave is in the $j$th column from the left end of the swarm, 
one of the two robots in the wave moves to the empty cell of the $(j+1)$st column 
so that the $j$th column represents $b_{j}$, which is stored at 
$copy[j-1]$. 
After the movement, the wave is in the $(j+1)$st column and 
the left $j$ columns of the swarm represent $b_1 b_2 \ldots b_j$. 
When the wave reaches the copy, it repeats the same procedure 
until it reaches the right end of the swarm. 
Specifically, when the wave reaches the column of $copy[0]$ 
(i.e., $\ell$th column from the left end of the swarm), 
$copy[\ell-1]$ specifies which of the two robots in the wave 
move to which cell because $\ell \geq 2$. 
Now, $B$ is fully stored in the left $\ell$ columns of the swarm and 
the wave proceeds to right with copying these columns. 
When $k$ is odd, the last robots moves to the $0$th row of the right 
end of the swarm. 
Then, the swarm retains the canonical state. 

\begin{figure}[tbp] 
   \centering 
    \includegraphics[width=3cm]{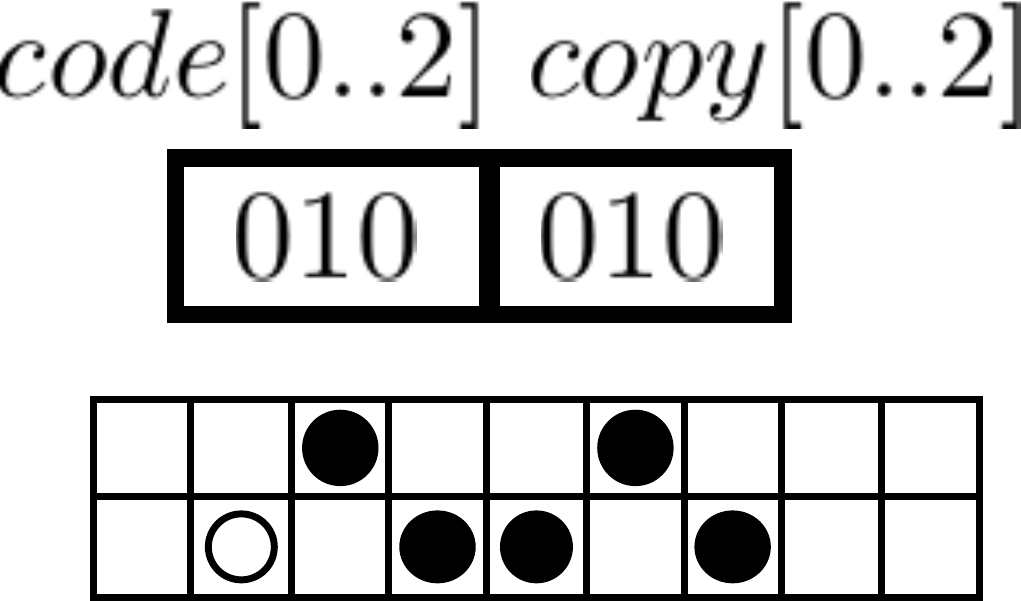}
    \caption{Codeword for $010$ in $ALG_2$. } 
 \label{fig:structureA2}
\end{figure}
\begin{figure}[tbp]
     \centering 
     \includegraphics[width=14cm]{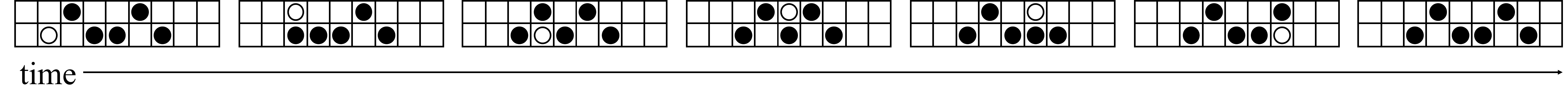} 
 \caption{Illustration of a right shift by one column realized by $ALG_2$. 
A robot represented by a black circle is stationary, while 
a one represented by a white circle moves and yields a new configuration.}
 \label{fig:shift-by-1}
\end{figure}

\begin{lemma}
\label{L331}
There is a code $\gamma$ based on algorithm $ALG_2$ whose size is $|\gamma|=2^{\lfloor k/2 \rfloor}$. 
The transmission delay $d_{ALG_2}$ is $k(m-k)$. 
\end{lemma}

\begin{proof}
The fact that $|\gamma|=2^{\lfloor k/2 \rfloor}$ is obvious from the definition 
of algorithm $ALG_2$. 

As for the transmission delay, 
one cycle of transmission shifts the canonical state right by one column. 
During the movement, the wave moves from the left end of the swarm 
to the right end 
with making one robot move in each column. 
Thus the total number of movement in a cycle is $k$. 
Consequently, $d_{ALG_2} = k(m-k)$. 
\end{proof}

By comparing two transmission delays in Lemmas\ref{L303} and \ref{L331},
when $m (\gg k)$ is large, $ALG_2$ is 10 times as fast as $ALG_1$.

\section{Variable swarm size code} 
\label{sec:varsizecode}

The information transmission scheme in Introduction allow us 
to assign configurations of different sizes to different 
symbols.
Suppose that, in a code $\gamma$,
a configuration with size $k_i$ is assigned to a symbol $s_i$ 
whose occurrence probability is $s_i$.
Let $K_{\gamma} = \sum_{i=1}^{\alpha} k_i p_i$ be the average swarm size of $\gamma$
and define $K^* = \min_{\gamma} K_{\gamma}$,
which is the optimal average swarm size.
This section analyzes bounds of $K^*$.  

\subsection{Lower bounds of $K^*$}

We derive a lower bound of $K^*$ on $G_8(m,n)$.
We can assume without loss of generality that 
$p_i \leq p_{i+1}$ for $i = 1, 2, \ldots , \alpha - 1$ 
and $k_i \leq k_{i+1}$ for $i = 1, 2, \ldots , \alpha - 1$,
since, if $k_i > k_{i+1}$ holds, 
we could reduce $K^*$ 
by exchanging the assignment to $s_i$ and $s_{i+1}$.
The size of any code realized by a swarm of $k$ robots on $G_8(m,n)$
is bounded from above by $2^{6(k-1)}$,
i.e., $\mu_{G_8(m,n)} \leq 2^{6(k-1)}$ by Lemma~\ref{L301}.
A lower bound of $K^*$ can be achieved by a code constructed
under the assumption that $2^{6(k-1)}$ configurations of size $k$ 
are assignable to symbols for each $k$.
Let $H(S) = - \sum_{i = 1}^{\alpha} p_i \log p_i$ be the entropy of $S$. 

\begin{theorem}
\label{T401}
Suppose that $m > k$ and $n\geq 2$.
Then for any code for $S$ on $G_8(m,n)$,
\begin{equation*}
 K ^* > \frac{1}{6} H(S) + 1 - \frac{1}{6} \log \lceil \log \frac{63 \alpha+1}{6} \rceil.
\end{equation*}
\end{theorem}

\begin{proof}
By assumption $k_i \leq k_{i+1}$ holds for $i = 1, 2, \ldots , \alpha - 1$.
We first calculate $k_{\alpha}$.
By assumption $k_{\alpha} = \max_i k_i$ holds.
Then $\sum_{k = 1}^{k_{\alpha}-1} 2^{6(k-1)} < \alpha \leq \sum_{k = 1}^{k_{\alpha}} 2^{6(k-1)}$.
By a simple calculation,
we have $k_{\alpha} = \left\lceil \log \frac{63 \alpha+1}{6} \right\rceil$.

Let $q_i = 2^{-6(k_i - 1)}/k_{\alpha}$.
Then, $q_i > 0$ for all $i = 1, 2,\ldots, \alpha$ and 
\[
\sum_{i = 1}^{\alpha} q_i = \sum_{k = 1}^{k_{\alpha}} \frac{2^{-6(k-1)}}{k_{\alpha}} 2^{6(k-1)} \leq 1.
\]
By Shannon's lemma, we have 
\begin{eqnarray*}
H(S) = - \sum_{i=1}^{\alpha} p_i \log p_i 
&<& - \sum_{i=1}^{\alpha} p_i \log q_i \\ 
&=& - \sum_{i=1}^{\alpha} p_i \log \frac{2^{-6(k_i-1)}}{k_{\alpha}} \\ 
&=& \sum_{i=1}^{\alpha} p_i \left(6(k_i-1) + \log k_{\alpha} \right) \\ 
&=& 6 \sum_{i=1}^{\alpha} p_i k_i -6 + \log k_{\alpha} \\ 
&=& 6 K^* -6 + \log k_{\alpha}. 
\end{eqnarray*}

Thus 
\begin{eqnarray*}
K^* &>& \frac{1}{6} H(S) +1 - \frac{1}{6} \log k_{\alpha} \\ 
&>& \frac{1}{6} H(S) + 1 - \frac{1}{6} \log \left\lceil \log \frac{63 \alpha+1}{6} \right\rceil.
\end{eqnarray*}
\end{proof}

Note that $\log \log \alpha > \frac{1}{6} \log \left\lceil \log \frac{63 \alpha+1}{6} \right\rceil - 1$, 
we have
\[
K^* > \frac{1}{6}(H(S) - \log \log \alpha).
\]

A bound of $K^*$ for any code on $G_8(m,2)$ can be derived
by using Lemma~\ref{L302}, by a similar discussion.

\begin{corollary}
\label{C401}
Suppose that $m > k$.
Then for any code for $S$ on $G_8(m,2)$,
\begin{equation*}
 K^* > 0.78 H(S) - 0.79 \log \left( 1 + 0.79 \log \alpha \right). 
\end{equation*}
\end{corollary}

Since $1 + 0.79 \log \alpha < \log \alpha$ when $\alpha \geq 2^5$,
$K^* > 0.78 H(S) - 0.79 \log \log \alpha$,
which may be approximated by $\frac{4}{5}(H(S) - \log \log \alpha)$
when $\alpha \geq 2^5$.
Note that regardless of $\alpha$, $H(S)$ can be 0, when $p_1 = 1$,
which means that these lower bounds can be negative and become meaningless.
On the other hand, when $p_i = 1/\alpha$ for all $i = 1, 2, \ldots , \alpha$,
$H(S)$ take the maximum value $\log \alpha$.
Thus $K^*$ on $TG_8(m,2)$ is roughly bounded from below by $0.78 \log \alpha$, 
when $\log \log \alpha$ is negligible.

\subsection{Upper bound of $K^*$ based on $ALG_1$} 

To derive an upper bound of $K^*$,
we construct a code $\gamma_{ALG_1}$ for $S$ based on $ALG_1$.
Although there are $2^{k-14}$ size $k$ configurations assignable to symbols for any $k > 14$,
there are no size $k$ configurations assignable to symbols for $k \leq 14$.
For each $k \leq 14$ however, we can easily construct a code with code size 1 
by using a trivial locomotion algorithm illustrated in Figure~\ref{figure:loco}.

\begin{figure}[tbp] 
 \centering 
 \includegraphics[width=11cm]{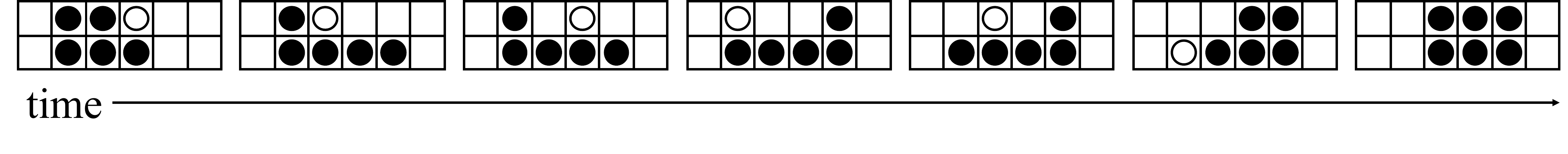}
 \caption{Locomotion of a small number of robots on $G_8(m,2)$}
 \label{figure:loco}
\end{figure}

Thus we use the configurations defined by $ALG_1$ as well as one configuration defined above 
for each $k \leq 14$.
Let $X$ be the set of all configurations assignable.
Code $\gamma_{ALG_1}$ is defined as follows:
We assign a configuration of size $k_i$ in $X$ to a symbol $s_i$ 
in such a way that $k_i \leq k_{i+1}$ holds for $i = 1, 2, \ldots , \alpha - 1$.
By definition, $k_{\alpha} \leq \alpha$.
Let $K_{ALG_1}$ denote the average code size of $\gamma_{ALG_1}$.

\begin{theorem}
\label{T402}
When $m > \alpha$,
$K_{ALG_1} < H(S) + 15$ holds.
\end{theorem}

\begin{proof}
Consider the ``tower'' of codewords of $\gamma_{ALG_1}$.
Each of the first $14$ levels contains a single codeword, 
and the $\ell$th layer contains $2^{\ell-14}$ codewords for $\ell \geq 15$.

Shannon's noiseless coding theorem guarantees the existence 
of a prefix-free code whose average code length is smaller than $H(S) + 1$.
Let $\ell_i$ be the length of the codeword of $s_i$ in such a prefix-free code. 
By the definition of $\gamma_{ALG_1}$, $k_i - 14 \leq \ell_i$. 
Thus
\begin{equation*}
K_{ALG_1} \leq \sum_{i=1}^{\alpha} p_i (\ell_i + 14) < H(S)+15. 
\end{equation*}
\end{proof}

Thus the performance of $\gamma_{ALG_1}$ in terms of the average code size is sufficiently 
good on $G_8(m,2)$, 
since, very roughly, $K_{\gamma}$ of any code $\gamma$ on $G_8(m,2)$ is bounded from below 
by $\frac{4}{5}(H(S) - \log \log \alpha)$.

Code $\gamma_{ALG_1}$ can be used not only on $G_8(m,2)$ but also on $G_8(m,n)$ for $n > 2$.
On $G_8(m,n)$, its performance may not be good, however, 
since our lower bound of $K^*$ on $G_8(m,n)$ is $\frac{1}{6}(H(S) - \log \log \alpha)$,
and there is still a large gap from $K_{ALG_1}$.

\subsection{Average transmission delay}

We next analyze the average transmission delay.
For any code $\gamma$ on $G_8(m,2)$, let $D_{\gamma} = \sum_{i=1}^{\alpha} d_i p_i$.
We are interested in the optimal average transmission delay $D^* = \min_{\gamma} D_{\gamma}$.
First we derive a lower bound of $D^*$.
Consider any code $\gamma$ on $G_8(m,n)$, where $m > 2 \alpha$.
We may assume without loss of generality that $\alpha \geq k_i$ for any $i = 1, 2, \ldots , \alpha$.
By Theorem~\ref{T204}, $d_i \geq k_i (m - 2k_i + 1)$.
Since $\alpha \geq k_i$,
\[
D_{\gamma} \geq (m+1) K_{\gamma} - 2 \sum_{i = 1}^{\alpha} k_i^2 p_i \geq (m + 1 - 2\alpha)K_{\gamma}.
\]
Since $K_{\gamma} \geq K^*$, we have:

\begin{observation}
\label{O401}
$D^* \geq (m + 1 - 2 \alpha) K^*$.
\end{observation}

To derive an upper bound of $D^*$, 
Let $D_{ALG_1}$ be the average delay of $\gamma_{ALG_1}$.
By Lemma~\ref{L303}, $d_i \leq (10k_i-123.5)(m-k_i)$.
Since $\alpha \geq k_i$ holds for all $i = 1, 2, \ldots , \alpha$,
we have:

\begin{observation}
\label{O402}
$D_{ALG_1} \leq 10 m K_{ALG_1}$.
\end{observation}

Thus $\gamma_{ALG_1}$ may not be a good choice from the view of 
the average transmission delay,
even considering the fact that $K_{ALG_1} \leq H(S) + 15$.
The average number of steps necessary to move one column is
approximated from above by $10 H(S)$ when $m \gg \alpha$.

Since algorithm $ALG_2$ has a smaller 
transmission delay,
we analyze a code $\gamma_{ALG_2}$ based on $ALG_2$.
Since $ALG_2$ is defined only for $k \geq 4$,
like $\gamma_{ALG_1}$, consider set $X$ that contains standard configurations
defined in $ALG_2$ for all $k \geq 4$,
as well as a single configuration of size $i$ for each of $i = 1, 2, 3$.
Like $\gamma_{ALG_1}$, $\gamma_{ALG_2}$ is defined as follows:
We assign a size $k_i$ configuration in $X$ to a symbol $s_i$ 
in such a way that $k_i \leq k_{i+1}$ holds for $i = 1, 2, \ldots , \alpha - 1$.
Let $K_{ALG_2}$ denote the average code size of $\gamma_{ALG_2}$.

\begin{theorem}
\label{T403}
Suppose that $m > \alpha > 1$.
Then $K_{ALG_2} < 2H(S)$. 
\end{theorem}

\begin{proof}
Recall that $p_i \geq p_{i+1}$ for $i = 1, 2, \ldots , \alpha - 1$.
Then $p_i \leq 1/i$, and hence $-\log p_i \geq \log i$. 

Consider the ``tower'' of codewords of $\gamma_{ALG_2}$.
Each of the first $3$ levels contains a single codeword, 
and the $\ell$th layer contains $2^{\lfloor \frac{k}{2} \rfloor}$ 
codewords for $k \geq 4$.
Code $\gamma_{ALG_2}$ packs the symbols in its order $s_1, s_2, \ldots, s_{\alpha}$
to codewords from the first level to the higher level in its order. 

We then divide the tower into three small towers, $T_1$, $T_2$, and $T_3$, 
where $T_1$ consists of the first 3 levels, 
$T_2$ consists of even levels greater than 3, 
and $T_3$ consists of odd levels greater than 4. 
Hence, $h$th level of $T_2$ (and $T_3$ also) corresponds to 
a binary tree in the sense that it contains $2^{h-1}$ codewords. 

First consider a codeword for $s_i$ in $T_2$. 
Let $h_i$ be the level that contains a configuration assigned to $s_i$ in $T_2$. 
The total number of symbols assigned to the level lower than 
$(h_i-1)$ in $T_1$, $T_2$, and $T_3$ is smaller than $i$. 
Thus we have $3 + 2(1+2+\cdots + 2^{h_i-1}) < i$,
which implies that $h_i < \log i - 1$.
Since the size of $k_i$ is $2h_i + 2$,
$k_i < 2 \log i$.
We then consider a codeword $s_i$ in $T_3$. 
Let $h_i$ be the level of $s_i$ in $T_3$. 
In the same way, $3 + 2(1+2+\cdots + 2^{m_i-1}) + 2^{m_i} < i$,
which implies $h_i < \log i - \log 3$.
Since the size of $k_i$ is $2h_i + 3$,
$k_i < 2 \log i + 3 - 2 \log 3 < 2 \log i$.
Thus, we have 
\begin{eqnarray*}
\sum_{i=1}^{\alpha} p_i k_i 
&<& \sum_{i=1}^{\alpha} 2 p_i \log i  \\ 
&=& 2 \sum_{i=1}^{\alpha} p_i \log i \\ 
&\leq& 2 \sum_{i=1}^{\alpha} p_i (- \log p_i) \\ 
&=& 2 H(S). 
\end{eqnarray*}
\end{proof}

Thus the upper bound of $K_{ALG_2}$ obtained by the theorem is roughly 
twice as much as that of $K_{ALG_1}$.
Let $D_{ALG_2}$ be the average transmission delay of $\gamma_{ALG_2}$.
By Lemma~\ref{L331}, $d_i = k_i (m - k_i)$.
Thus
\begin{observation}
\label{O402}
$D_{ALG_2} \leq m K_{ALG_2}$.
\end{observation}

Since $K_{ALG_2} < 2 H(S)$,
we have $\leq D^* \leq 2m H(S)$.
The average number of steps necessary to move one column is
approximated from above by $2 H(S)$ when $m \gg \alpha$.
Thus 
the variable swarm size code by 
$ALG_2$ is 5 times as fast as that by $ALG_1$.

\section{Conclusion}
\label{sec:concl}

We proposed an information transmission scheme by a swarm of 
anonymous oblivious mobile robots on a graph. 
We mainly analyzed the performance of our scheme in terms of 
code size and transmission delay in the 8-gird 
and proposed two algorithms one achieves exponential code size 
with large transmission delay and the other achieves 
optimal transmission delay with small code size. 
We finally extended these algorithms for 
variable swarm size codes. 

There are many open problems related to the proposed scheme. 
First, we could not find any algorithm with optimal code size and 
optimal transmission delay. 
Second, the gap between the upper bound and the lower bound 
of the (expected) swarm size of the fixed swarm size code and 
of the variable swarm size code needs to be closed. 
One approach is a more sophisticated technique to upper bound 
the number of initial configurations and terminating behaviors. 
Third, parallel movement of robots might speed up the transmission 
and makes algorithms simpler. 

One of the most important future directions is robustness. 
We put our basis on Shannon's noiseless coding theorem, and 
the next step is to consider faulty robots. 
We believe Shannon's noisy channel coding theorem help the investigation. 
Another direction is to investigate local algorithms 
that restricts the visibility of the robots to a constant distance.

\bibliographystyle{plain}
\bibliography{papers}

\begin{thebibliography}{10}

\bibitem{AFSY08}
Yuichi Asahiro, Satoshi Fujita, Ichiro Suzuki, and Masafumi Yamashita.
\newblock A self-stabilizing marching algorithm for a group of oblivious
  robots.
\newblock In {\em Proceedings of the 12th International Conference on
  Principles of Distributed Systems (OPODIS 2008)}, pages 125--144, 2008.

\bibitem{BMPT11}
Fran{\c{c}}ois Bonnet, Alessia Milani, Maria Potop-Butucaru, and S{\'e}bastien
  Tixeuil.
\newblock Asynchronous exclusive perpetual grid exploration without sense of
  direction.
\newblock In {\em Proceedings of the 15th International Conference on
  Principles Of Distributed Systems (OPODIS 2011)}, pages 251--265, 2011.

\bibitem{CP07}
Davide Canepa and Maria~Gradinariu Potop-Butucaru.
\newblock Stabilizing flocking via leader election in robot networks.
\newblock In {\em Proceedings of the 9th International Conference on
  Stabilization, Safety, and Security of Distributed Systems (SSS 2007)}, pages
  52--66, 2007.

\bibitem{CFPS12}
Mark Cieliebak, Paola Flocchini, Giuseppe Prencipe, and Nicola Santoro.
\newblock Distributed computing by mobile robots: Gathering.
\newblock {\em SIAM Journal on Computing}, 41:829--879, 2012.

\bibitem{CDPIM10}
Julien Clement, Xavier D\'efago, Maria~Gradinariu Potop-Butucaru, Taisuke
  Izumi, and Stephane Messika.
\newblock The cost of probabilistic agreement in oblivious robot networks.
\newblock {\em Information Processing Letters}, 110(11):431--438, 2010.

\bibitem{CT06}
Thomas~M. Cover and Joy~A. Thomas.
\newblock {\em Elements of Information Theory}.
\newblock John Wiley \& Sons, Inc., 2nd edition, 2006.

\bibitem{DFSY15}
Shantanu Das, Paola Flocchini, Nicola Santoro, and Masafumi Yamashita.
\newblock Forming sequences of geometric patterns with oblivious mobile robots.
\newblock {\em Distributed Computing}, 28(2):131--145, 2015.

\bibitem{DLPRT12}
St{\'e}phane Devismes, Anissa Lamani, Franck Petit, Pascal Raymond, and
  S{\'e}bastien Tixeuil.
\newblock Optimal grid exploration by asynchronous oblivious robots.
\newblock In {\em Proceedings of the 14th International Symposium on
  Stabilization, Safety, and Security of Distributed Systems (SSS 2012)}, pages
  64--76, 2012.

\bibitem{DFSV18}
Giuseppe~A. Di~Luna, Paola Flocchini, Nicola Santoro, and Giovanni Viglietta.
\newblock Turingmobile: A {T}uring machine of oblivious mobile robots with
  limited visibility and its applications.
\newblock In {\em Proceedings of the 32nd International Symposium on
  Distributed Computing (DISC 2018)}, pages 19:1--19:18, 2018.

\bibitem{DP07}
Yoann Dieudonn\'e and Franck Petit.
\newblock Robots and demons (the code of the origins).
\newblock In {\em Proceedings of the 4th International Conference on Fun with
  Algorithms (FUN2007)}, pages 108--119, 2007.

\bibitem{DPV10}
Yoann Dieudonn{\'e}, Franck Petit, and Vincent Villain.
\newblock Leader election problem versus pattern formation problem.
\newblock In {\em Proceedings of the 24th International Symposium on
  Distributed Computing (DISC2010)}, pages 267--281, 2010.

\bibitem{D74}
Edsger~W. Dijkstra.
\newblock Self-stabilizing systems in spite of distributed control.
\newblock {\em Communications of the ACM}, 17(11):643--644, 1974.

\bibitem{DYKY18}
Keisuke Doi, Yukiko Yamauchi, Shuji Kijima, and Masafumi Yamashita.
\newblock Exploration of finite 2{D} square grid by a metamorphic robotic
  system.
\newblock In {\em Proceedings of the 20th International Symposium on
  Stabilization, Safety, and Security of Distributed Systems (SSS 2018)}, pages
  96--110, 2018.

\bibitem{FIPS13}
Paola Flocchini, David Ilcinkas, Andrzej Pelc, and Nicola Santoro.
\newblock Computing without communicating: Ring exploration by asynchronous
  oblivious robots.
\newblock {\em Algorithmica}, 65(3):562--583, 2013.

\bibitem{FPS12}
Paola Flocchini, Giuseppe Prencipe, and Nicola Santoro.
\newblock {\em Distributed Computing by Oblivious Mobile Robots}.
\newblock Morgan \& Claypool Publishers, 2012.

\bibitem{FPS19}
Paola Flocchini, Giuseppe Prencipe, and Nicola Santoro.
\newblock {\em Distributed Computing by Mobile Entities, Current Research in
  Moving and Computing}.
\newblock Lecture Notes in Computer Science 11340. Springer, 2019.

\bibitem{FPSW08}
Paola Flocchini, Giuseppe Prencipe, Nicola Santoro, and Peter Widmayer.
\newblock Arbitrary pattern formation by asynchronous, anonymous, oblivious
  robots.
\newblock {\em Theoretical Computer Science}, 407:412--447, 2008.

\bibitem{FYOKY15}
Nao Fujinaga, Yukiko Yamauchi, Hirotaka Ono, Shuji Kijima, and Masafumi
  Yamashita.
\newblock Pattern formation by oblivious asynchronous mobile robots.
\newblock {\em SIAM Journal on Computing}, 44:740--785, 2015.

\bibitem{GP04}
Vincenzo Gervasi and Giuseppe Prencipe.
\newblock Coordination without communication: the case of the flocking problem.
\newblock {\em Discrete Applied Mathematics}, 144:324--344, 2004.

\bibitem{TKGT18}
Taisuke Izumi, Daichi Kaino, Maria Gradinariu Potop-Butucaru, and S\'{e}bastien
  Tixeuil.
\newblock On time complexity for connectivity-preserving scattering of mobile
  robots.
\newblock {\em Theoretical Computer Science}, 738:42--52, 2018.

\bibitem{IKY14}
Tomoko Izumi, Sayaka Kamei, and Yukiko Yamauchi.
\newblock Approximation algorithms for the set cover formation by oblivious
  mobile robots.
\newblock In {\em Proceedings of the 18th International Conference on
  Principles of Distributed Systems (OPODIS 2014)}, pages 233--247, 2014.

\bibitem{LYKY18}
Zhiqiang Liu, Yukiko Yamauchi, Shuji Kijima, and Masafumi Yamashita.
\newblock Team assembling problem for asynchronous heterogeneous mobile robots.
\newblock {\em Theoretical Computer Science}, 721:27--41, 2018.

\bibitem{PPV15}
Linda Pagli, Giuseppe Prencipe, and Giovanni Viglietta.
\newblock Getting close without touching: near-gathering for autonomous mobile
  robots.
\newblock {\em Distributed Computing}, 28(5):333--349, 2015.

\bibitem{SY99}
Ichiro Suzuki and Masafumi Yamashita.
\newblock Distributed anonymous mobile robots: Formation of geometric patterns.
\newblock {\em SIAM Journal on Computing}, 28(4):1347--–1363, 1999.

\bibitem{YS10}
Masafumi Yamashita and Ichiro Suzuki.
\newblock Characterizing geometric patterns formable by oblivious anonymous
  mobile robots.
\newblock {\em Theoretical Computer Science}, 411:2433--2453, 2010.

\bibitem{YUKY17}
Yukiko Yamauchi, Taichi Uehara, Shuji Kijima, and Masafumi Yamashita.
\newblock Plane formation by synchronous mobile robots in the three-dimensional
  euclidean space.
\newblock {\em Journal of the ACM}, 64(3):16:1--16:43, 2017.

\bibitem{YUKY16}
Yukiko Yamauchi, Taichi Uehara, and Masafumi Yamashita.
\newblock Brief announcement: Pattern formation problem for synchronous mobile
  robots in the three dimensional euclidean space.
\newblock In {\em Proceedings of the 35th ACM Symposium on Principles of
  Distributed Computing (PODC 2016)}, pages 447--449, 2016.

\bibitem{YY14}
Yukiko Yamauchi and Masafumi Yamashita.
\newblock Randomized pattern formation algorithm for asynchronous oblivious
  mobile robots.
\newblock In {\em Proceedings of the 28th International Symposium on
  Distributed Computing (DISC 2014)}, pages 137--151, 2014.

\bibitem{YXCD08}
Yan Yang, Naixue Xiong, Nak~Young Chong, and Xavier D\'{e}fago.
\newblock A decentralized and adaptive flocking algorithm for autonomous mobile
  robots.
\newblock In {\em Proceedings of the 3rd International Conference on Grid and
  Pervasive Computing -- Workshops}, pages 262--268, 2008.

\end{thebibliography}

\end{document}